\documentclass[11pt]{llncs}
\usepackage{multirow}
\usepackage{enumitem}
\usepackage{amssymb}
\usepackage{amsmath}
\usepackage{mathtools}
\begin{document}
\title{Privacy-preserving Weighted Federated Learning within Oracle-Aided MPC Framework}
\author{Huafei Zhu, Zengxiang Li, Mervyn Cheah, Rick Siow Mong Goh}
\institute{IHPC, A*STAR, Singapore}
\maketitle

\begin{abstract}
This paper studies privacy-preserving weighted federated learning within the oracle-aided multi-party computation (MPC) framework. Our contribution mainly comprises the following three-fold:
\begin{itemize}
\item In the first fold, a new notion which we call weighted federated learning (wFL) is introduced and formalized. The weighted federated learning concept formalized in this paper differs from that presented in the McMahan et al.'s paper since both addition and multiplication operations are executed over cipher space in our model while these operations are executed over plaintext space in McMahan et al.'s model;

\item In the second fold, an oracle-aided MPC solution for computing weighted federated learning is formalized by decoupling the security of the defined weighted federated learning system from that of the underlying multi-party computation. Our decoupling formulation may benefit machine learning developers to select their best security practices from the state-of-the-art secure MPC tool sets;

\item In the third fold, a concrete solution to the weighted federated learning problem is presented and analysed. The security of our implementation is guaranteed by the security composition theorem assuming that the underlying multiplication algorithm is secure against honest-but-curious adversaries.  
\end{itemize}


\end{abstract}

\section{Introduction}
The concept of Federated Learning (FL) first introduced by McMahan et al.~\cite{Mahan1701} is a decoupling of model training from the need for direct access to the raw training data. A formal definition of Federated Learning later has been formalized by Qiang Yang et al.\cite{YQ1901}, where datasets defined in the FL framework are categorized as horizontal, vertical and hybrid types. Roughly speaking, in the horizontal FL, datasets of different organizations have same feature space but little intersection on the sample space~\cite{Mahan1702,He1901}; In the vertical FL, datasets of different organization have same sample space (entity) but little intersection on the feature space; In the hybrid FL, feature spaces and sample spaces are overlapped in an non-negligible level~\cite{Hardy1701,Hardy1801}. We refer to the reader~\cite{Mahan1901,Mahan1902,YQ1902,Agrawal2001,Agrawal2002}(and the references therein) for further reference.

\subsection{The motivation problem}
Going through the FederatedAveraging algorithm presented in~\cite{Mahan1701} that works over the horizontal datasets, we know that each client $k$ locally computes $n_k$ number of local data samples for the local model $w_{t+1} ^k$ at the current $(t+1)$-round. The parameters $n_k$ and $w_{t+1}^k$ are then sent to the global FL server who in turn, computes the weighted average of the resulting model $w_{t+1}$ $\leftarrow$ $\sum_{k=1} ^K \frac{n_k}{n} w_{t+1} ^k$ where $K$ is the number of clients and $n $ = ${n_1} + \cdots + n_K$. From the client point of views, it is desirable both $n_k$ and $w_{t+1}^k$ are well protected since the variables contain sensitive information closely related to the client $k$. In fact, a demonstrative attack sketched in~\cite{Mahan1701} shows that if the update is the total gradient of the loss on all of local data, and the features are a sparse bag-of-words, then the non-zero gradient reveals exactly which words the user has entered on the device. 

Since $n_k$ is the number of local data samples for the local model $w_{t+1} ^k$, we may simply map these parameters in the context of the FederatedAveraging to the standard notion of weight and feature pair ($n_k$, $w_{t+1} ^k$) in the context of the machine learning framework, where $n_k$ stands for weight and $w_{t+1} ^k$ stands for the feature at the $(t+1)$-round. A naive solution to protect users' data could be that, to keep $n_k$ and $w_{t+1}^k$ private, the client (data contributor) $k$ could first encrypt $n_k$ and $w_{t+1} ^k$ and then send the resulting ciphertexts $[n_k]$ and $[w_{t+1}^k]$ to the global server. The selection of the underlying encryption scheme that is used to encrypt $n_k$ and $w_{t+1} ^k$ is flexible which can be a secret sharing scheme (either Shamir secret sharing or additively secret sharing) based encryption or a homomorphic cryptosystem (e.g., additively homomorphic encryption or multiplicative encryption or (somewhat) fully homomorphic encryption). 

Roughly speaking, a FederatedAveraging algorithm working over cipher space is called weighted federated learning (wFL) since both $n_k$ and $w_{t+1} ^k$ are encrypted and thus are unknown to the global server. We stress that the notion of FederatedAveraging algorithm works over plaintexts while the notion of wFL works over ciphertexts. Given encrypted weight and feature pairs, the global server then performs the following computations over ciphers:
\begin{itemize}
\item Computing the summation over $([n_1]$, $\cdots$, $[n_K])$ such that $[n]$ = $[n_1]$ + $\cdots$ + $[n_K]$;

\item Computing $\mathrm{[WeightedAggregating]}$ = $[n_1]$ $[w_{t+1} ^1]$ + $\cdots$ +  $[n_K][w_{t+1} ^k]$;

\item Decrypting $[n]$ and $\mathrm{[WeightedAggregating]}$ to obtain the corresponding plaintexts of parameters $n$ and $\mathrm{WeightedAggregating}$;

\item Updating the global parameter $\mathrm{FederatedAveraging}$ by computing $\mathrm{WeightedAggregating}/n$ (possibly, multi iterations will be conducted depending a pre-defined threshold for this training model).
\end{itemize}

Since the suggested solution to wFL comprises two basic addition and multiplication arithmetic operations defined over ciphers, the state-of-the-art secure multi-party computation (MPC) platforms and secure machine learning (ML) platforms can be applied to solve the $\mathrm{wFL}$ problem. However, a direct application of additive or multiplicative or (somewhat) fully-homomorphic encryption to the above problem could result in an inefficient solution since in the federated learning scenario, the number of total sample data of an application is big (an experiment for 100 clients each with 600 data samples has been demonstrated by McMahan et al.~\cite{Mahan1701}). 

Recall that the challenging of MPC based on SPDZ framework~\cite{Ivan2012,Ivan2013,Ivan2018,Smart1901,Smart2001} is to generate Beaver multiplication triple set efficiently and securely. Three methods are known so far to generate Beaver triple: 1) somewhat fully-homomorphic based solution; 2) Trusted Third Party based software solution and 3)Enclave based hardware solution. The challenging of zero-summation based MPC such as ShareMind~\cite{Dan2012,Lindell201601} is the scalability problem since it is inherently suitable for 3-party computation. To resolve the scalability of the zero-summation based MPC, we can apply the committee selecting technique presented in~\cite{Micali18,Micali19} and thus ShareMind may be more suitable for weighted federated learning solutions.  

\subsection{This work}
Generalizing the above observation, we are able to introduce the notion of weighted federated learning defined over cipher space, which is stated informally below:

\subsubsection{Weighted Federated Learning (wFL)}: Let $P_1$, $\cdots$, $P_m$ be $m$ clients. Each client $P_i$ has its private input $inp_i$ =$(w_i, f_i)$ and outsources its encrypted data $([w_i], [f_i])$ to a set of FL computing servers, where $([w_i], [f_i])$ stands for a pair of encrypted weight $[w_i]$ and feature $[f_i]$. W.l.o.g, we simply assume that both $w_i$ and $f_i$ are integers. Let $\mathrm{WA}$$(inp_1,\cdots, inp_m)$ = $1/([w_1]+ \cdots + [w_m])$ $\times$ $([w_1] [f_1] + \cdots + [w_m] [f_m])$ be a machine learning mechanism maintained and managed by a global server for conducting the weighted aggregating algorithm $\mathrm{WA}$ whose input is $([w_i], [f_i])$ ($i=1, \cdots, m$) and output is $1/([w_1]+\cdots+ [w_m]) \times ([w_1] [f_1] +\cdots+ [w_m] [f_m])$.

\subsubsection{The challenging and solution}: As noted above, there are known solutions to the basic addition and multiplication arithmetic operations defined over ciphers. The evolution of the existing algorithms and protocols for implementing arithmetic operations defined over ciphers leaves us a challenging task $-$ how to evaluate the security of a federated learning system constructed from the evolving implementations. For example, in ShareMind, the multiplication operator based on the Du and Atallah's method~\cite{Du2001} was replaced by a newly developed zero-summation triple mechanism~\cite{Dan2012,Lindell201601}. To solve this challenging problem, we decouple the security of the weighted federated learning from that of the underlying arithmetic operations by viewing a known implementation of arithmetic operation, or a protocol defined over ciphers as an oracle-aided computation; We then evaluate the security of the weighted federated learning system in the MPC protocol composition model. Our decoupling formulation may benefit machine learning developers to select their best security practices from the state-of-the-art security tool sets.
 
\subsubsection{The roadmap}: The rest of this paper is organized as follows: In section~2, syntax and security definition for weighted federated learning is introduced and formlaized; An efficient implementation and security proof are presented in Section~3. We conclude our work in Section~4.

\section{Syntax and security definition}
In this section, we are going to provide a formal definition for weighted Federated Learning and then define the security of wFL within the oracle-aided multi-party computation framework.

\subsection{Syntax of weighted federated learning}

\begin{definition}
A weighted Federated Learning protocol (wFL) consists of a group of clients ($c_1, \cdots, c_m$), a global Federated Learning server sFL and a group of MPC servers ($P_1 \cdots, P_n$). Each client $c_i$ holds a weight and feature pair $(x_i, y_i) \in Z_p ^* \times Z_p ^*$ ($p$ is a prime number) which is additively shared among MPC servers where $P_j$ holds $(x_{i,j}, y_{i,j})$, $x_i$ = $x_{i,1}+ \cdots + x_{i,n}$ and $y_i$ =$y_{i,1}+ \cdots + y_{i,n}$. By ($[x_i]$,$[y_i]$), we denote a pair of random shares $(x_{i,1}, \cdots, x_{i,n})$ and $(y_{i,1}, \cdots, y_{i,n})$ among $P_j$ ($j=1, \cdots, n$). The global federated learning server sFL defines a machine learning algorithm WA whose input is $([x_1], [y_1])$, $\cdots$, $([x_n], [y_n])$ and the output is the plaintext of the aggregation ($\sum _{k=1} ^K [x_i] \times [y_i]$, $\sum _{k=1} ^K [x_i]$).
\end{definition}

\begin{remark}
Please note that the definition of wFL presented in this paper, is different from that presented in the McMahan et al.'s paper since both addition and multiplication operations are executed over cipher space in our model while these operations are executed over plaintext space in McMahan et al.'s model~\cite{Mahan1701}.
\end{remark}

\begin{remark}
Please also note that the definition of wFL presented in this paper, is different from that presented in the Bonawitz et al.'s paper since both addition and multiplication operations are executed over cipher space in our model while ONLY addition operation is executed over cipher space in Bonawitz et al.'s model~\cite{Mahan1702} where each data contributor's weight is a public value.

\end{remark}

\subsection{Security definition of weighted federated learning}
The security of wFL protocol is formalized in the context of an oracle-aided secure multi-party computation (MPC) which in essence, is a decoupling of machine learning algorithm from the need for MPC that may benefit machine learning developers to select their best security practices from the state-of-the-art security tool sets. We briefly describe the notations and notions related to oracle-aided secure multi-party computation below and refer to the reader~\cite{Goldreichbook1,Goldreichbook2}) for more details.

Let $f:$ $(\{0, 1\}^{*})^m $ $\rightarrow$ $(\{0, 1\}^{*})^m $ be an $m$-ary functionality, where $f_i(x_1, \cdots, x_m)$ denotes the $i$th element of $f(x_1, \cdots, x_m)$. Let $[m]$= $\{1, \cdots, m\}$, and for $I \in \{i_1, \cdots, i_t\}$ $\subseteq$ $[m]$, we let $f_I(x_1, \cdots, x_m)$ denote the subsequence $f_{i_1}(x_1, \cdots, x_m)$, $\cdots$, $f_{i_t}(x_1, \cdots, x_m)$. Let $\mathrm{\Pi}$ be an $m$-party protocol for computing $f$. The view of the $i$-th party during an execution of $\mathrm{\Pi}$ on $\overline{x}$:= $(x_1, \cdots, x_m)$ is denoted by $\mathrm{View_i ^{\Pi}} (\overline{x})$. For $I$ = $\{i_1, \cdots, i_t\}$, we let $\mathrm{View_I ^{\Pi}} (\overline{x})$:= ($I$, $\mathrm{View_{i_1} ^{\Pi}} (\overline{x})$, $\cdots$, $\mathrm{View_{i_t} ^{\Pi}} (\overline{x})$). In case $f$ is a deterministic $m$-ary functionality, we say $\mathrm{\Pi}$ privately computes $f$ if there exists a probabilistic polynomial-time algorithm denoted $S$, such that for every $I \subseteq [m]$, it holds that $S(I$, $(x_{i_1}, \cdots, x_{i_t})$, $f_I(\overline{x}))$ is computationally indistinguishable with $\mathrm{View_I ^{\Pi}} (\overline{x})$. In general case, $S(I$, $(x_{i_1}, \cdots, x_{i_t})$, $f_I(\overline{x})$, $f(\overline{x}))$ is computationally indistinguishable with $\mathrm{View_I ^{\Pi}}$ $((\overline{x})$, $f(\overline{x}))$.

An oracle-aided protocol is a protocol augmented by a pair of oracle types, per each party. An oracle-call step is defined as follows: a party writes an oracle request on its own oracle tape and then sends it to the other parties; in response, each of the other parties writes its query on its own oracle tape and responds to the first party with an oracle call message; at this point the oracle is invoked and the oracle answer is written by the oracle on the ready-only oracle tape of each party. An oracle-aided protocol is said to privately reduce $g$ to $f$ if it securely computes $g$ when using the oracle-functionality $f$. In such a case, we say that $g$ is securely reducible to $f$.

\begin{definition}
An multiplication-oracle aided $wFL$ is privacy-preserving if $wFL$ is privately reducible to the multiplication functionality.
\end{definition}

\begin{remark}
Please notice that we do not provide the privacy-preserving reduction to the addition oracle since the underlying data sharing scheme is an additively secret sharing.
\end{remark}

\section{The implementation and security proof}
In this section, a concrete solution of wFL based on the additive data sharing with the help of the zero-summation technique defined over three-server setting is presented and analyzed. The security of our implementation is derived from the security composition theorem assuming that the underlying ShareMind Multiplication algorithm is secure against honest-but-curious adversaries.

\subsection{The implementation}
Our implementation consists of following steps: the data splitting, the resharing, the addition and the multiplication. Each of steps is depicted in details below: 

\subsubsection{The data splitting}

Suppose a wFL client Alice holds private data $[x]$ and $[y]$ locally. W.l.o.g., we assume that there are three MPC servers managed and maintained by independent computing service providers such as FL auditor ($P_1$), FL insurance company $P_2$) and FL client association ($P_3$). We assume that there is a secure (private and authenticated) channel between client Alice and each of MPC service providers. This assumption is standard and can be easily implemented under the standard PKI assumption. For simplicity, we assume that $x, y \in Z_p ^{*}$, where $p$ is a suitable large prime number (e.g., $|p|$ =512). The splitting procedure is defined below

\begin{itemize}
\item Alice selects $x_1, x_2 \in Z^*_p$ uniformly at random, and then sends $x_1$ to $P_1$, $x_2$ to $P_2$;

\item Alice computes $x_3 = x - x_2 -x_3$~mod~p and sends $x_3$ to $P_3$.
\end{itemize}
The splitting of the data $x$ is defined by $[x]$ = $(x_1, x_2, x_3)$ (as usual, a random split of data is also called an encryption of that data). Similarly, an encryption of $y$ is defined by $[y]$ = $(y_1, y_2, y_3)$, where $P_i$ holds $y_i$ ($i$= 1,2,3).

\subsubsection{The resharing}

A refreshing procedure is called whenever a multiplication operation is executed. The refreshing procedure is defined among $P_1$ (with input $x_1$), $P_2$ (with input $x_2$ ) and $P_3$ (with input $x_3$) such that $x$ = $x_1$ + $x_2$ + $x_3$:

\begin{itemize}
\item $P_1$ selects $r_1 \in Z_p ^*$ uniformly at random and sends $r_1$ to $P_2$ via a pre-defined secure channel;

\item Similarly, $P_2$ (resp. $P_3$) selects $r_2 \in Z_p ^*$ (resp. $r_3 \in_U Z_p ^*$) uniformly at random and sends $r_2$ (resp. $r_3$) to $P_3$ (resp. $P_1$) via a pre-defined secure channel;

\item $P_1$ locally computes $\sigma_1$ = $r_1$ - $r_3$~mod~$p$  and  $x_1 '$ = $x_1 + \sigma_1$~mod~$p$; $P_2$ locally computes $\sigma_2$ = $r_2$ - $r_1$~mod~$p$  and  $x_2 '$ = $x_2 + \sigma_2$~mod~$p$; $P_3$ locally computes $\sigma_3$ = $r_3$ - $r_2$~mod~$p$  and  $x_3 '$ = $x_3 + \sigma_3$~mod~$p$.

\end{itemize}
A refresh of $[x]$ is denoted by $[x']$ = $([x_1 '],[x_2 '],[x_3 '])$. One can verify that $x_1 ' + x_2 ' + x_3 '$~mod~$p$ = $x_1 + x_2 + x_3$~mod~$p$. 

\subsubsection{The addition}
Suppose $P_i$ holds shares of $x_i$ and $y_i$. $P_i$ locally computes $z_i$ = $x_i$ + $y_i$~mod~$p$ and then sends $z_i$ to the FL global server who computes $z_1 + z_2 + z_3$~mod~$p$ and thus gets the value of addition $x + y$~mod~$p$.

\subsubsection{The multiplication}
On input $(x_i, y_i)$, each of participants $P_i$ can jointly run the resharing protocol to get $(x_i', y_i')$ ($i =1, 2, 3$). The role of resharing protocol plays a one-time padding of shares. $P_i$ then sends its shares $(x_i', y_i')$ to $P_{i~mod~3 + 1}$. Then $P_1$ computes $z_1$ = $(x_1' y_1 ' + x_1 ' y_3 ' + x_3 ' y_1')$~mod~$p$; $P_2$ computes $z_2$ = $(x_2 ' y_2' + x_2 ' y_1' + x_1' y_2')$~mod~$p$ and $P_3$ computes $z_3$ = $x_3' y_3' + x_3' y_2' + x_2' y_3'$~mod~$p$. One can verify that $z_1 + z_2 + z_3$~mod~$p$ = $[x]$$[y]$~mod~$p$. 

\subsubsection{Putting things together}
Given an encryption of the weight and feature vectors $[n]$ =$([n_1], \cdots, [n_K])$ and $[w]$ =$([w_1], \cdots, [w_K])$, where $[n_k]$ = ($n_{k,1}$, $n_{k,2}$, $n_{k,3}$) and $[w_k]$ = ($w_{k,1}$, $w_{k,2}$, $w_{k,3}$). Notice that ($n_{k,1}$, $w_{k,1}$) is a secret share held by $P_1$, ($n_{k,2}$, $w_{k,2}$) is a share held by $P_2$ and $P_3$ holds ($n_{k,3}$, $w_{k,3}$) for $k= 1, \cdots, K$. Applying the addition and multiplication operations described above, we are able to solve the wFL problem.

\subsection{The proof of security}

\begin{theorem}
Let $\mathrm{g_{wFL}}$ be a weighted Federated Learning functionality defined in the three-server framework. Let $\mathrm{\Pi^{g_{wFL}|f_{mult}}}$ be an oracle-aided protocol that privately reduces $\mathrm{g_{wFL}}$ to $\mathrm{f_{mult}}$ and $\mathrm{\Pi^{f_{mult}}}$ be a protocol privately computes $\mathrm{f_{mult}}$. Suppose $\mathrm{g_{wFL}}$ is privately reducible to $\mathrm{f_{mult}}$ and that there exists a protocol for privately computing $\mathrm{f_{mult}}$, then there exists a protocol for privately computing $\mathrm{g_{wFL}}$.
\end{theorem}

\begin{proof}
We construct a protocol $\mathrm{\Pi}$ for computing $\mathrm{g_{wFL}}$. That is, we replace each invocation of the oracle $\mathrm{f_{mult}}$ by an execution of protocol $\mathrm{\Pi^{f_{mult}}}$. Note that in the semi-honest model, the steps executed $\mathrm{\Pi^{g_{wFL}|f_{mult}}}$ inside $\mathrm{\Pi}$ are independent the actual execution of $\mathrm{\Pi^{f_{mult}}}$ and depend only on the output of $\mathrm{\Pi^{f_{mult}}}$. 

For each $i=1,2,3$, let $\mathrm{S_i ^{g_{wFL}|f_{mult}}}$ and $\mathrm{S_i ^{f_{mult}}}$ be the corresponding simulators for the view of party $P_i$. We construct a simulator $S_i$ for the view of party $P_i$ in $\mathrm{\Pi}$. That is, we first run $\mathrm{S_i ^{g_{wFL}|f_{mult}}}$ and obtain the simulated view of party $P_i$ in $\mathrm{{\Pi}^{g_{wFL}|f_{mult}}}$. This simulated view includes queries made by $P_i$ and the corresponding answers from the oracle. Invoking $\mathrm{S_i ^{f_{mult}}}$ on each of partial query-answer $(q_i, a_i)$, we fill in the view of party $P_i$ for each of these interaction of $\mathrm{S_i ^{f_{mult}}}$. The rest of the proof is to show that $S_i$ indeed generates a distribution that is indistinguishable from the view of $P_i$ in an actual execution of $\mathrm{\Pi}$. 

Let $\mathrm{H_i}$ be a hybrid distribution represents the view of $P_i$ in an execution of $\mathrm{{\Pi}^{g_{wFL}|f_{mult}}}$ that is augmented by the corresponding invocation of $\mathrm{S_i ^{f_{mult}}}$. That is, for each query-answer pair $(q_i, a_i)$, we augment its view with $\mathrm{S_i ^{f_{mult}}}$. It follows that $\mathrm{H_i}$ represents the execution of protocol $\mathrm{\Pi}$ with the exception that $\mathrm{\Pi ^{f_{mult}}}$ is replaced by simulated transcripts. We will show that 

\begin{itemize}
\item the distribution between $\mathrm{H_i}$ and $\mathrm{\Pi}$ are computationally indistinguishable: notice that the distributions of $\mathrm{H_i}$ and $\mathrm{\Pi}$ differ $\mathrm{\Pi ^{f_{mult}}}$ and $\mathrm{S_i ^{f_{mult}}}$ which is computationally indistinguishable assuming that $\mathrm{\Pi ^{f_{mult}}}$ securely computes $\mathrm{f_{mult}}$. 
 
\item the distribution between $\mathrm{H_i}$ and $\mathrm{S_i}$ are computationally indistinguishable: notice that the distributions between ($\mathrm{{\Pi}^{g_{wFL}|f_{mult}}}$, $\mathrm{S_i ^{f_{mult}}}$) is computationally indistinguishable from ($\mathrm{{S_i}^{g_{wFL}|f_{mult}}}$, $\mathrm{S_i ^{f_{mult}}}$). The distribution ($\mathrm{{S_i}^{g_{wFL}|f_{mult}}}$, $\mathrm{S_i ^{f_{mult}}}$) defines $\mathrm{S_i}$. That means $\mathrm{H_i}$ and $\mathrm{S_i}$ are computationally indistinguishable.
\end{itemize}
\end{proof}

\begin{corollary}
Assuming that the underlying multiplication algorithm presented in~\cite{Dan2012} is secure against honest-but-curious adversary, our implementation is secure against the same adversarial type. 
\end{corollary}

\section{Conclusion}

In this paper, a new notion which we call weighted federated learning problem is introduced and formalized. The security of wFL is defined within the Oracle-aided MPC framework. An efficient solution to the wFL is implemented within the framework of ShareMind and we are able to show that if the underlying multiplication algorithm is secure against honest-but-curious adversary, then our implementation is secure against the same adversarial type.

\end{document}